\documentclass[12pt]{article}

\usepackage{enumitem}
\usepackage{comment}
\usepackage{amsthm}
\usepackage{amsmath}
\usepackage{amssymb}
\usepackage{amsfonts}
\usepackage{graphicx}
\usepackage{mathrsfs}
\usepackage{bm}
\usepackage{xspace}
\usepackage{subfig}

\newtheorem{theorem}{Theorem}[section]

\newtheorem{corollary}[theorem]{Corollary}
\newtheorem{openproblem}[theorem]{Open Problem}

\newcommand{\Pos}{\ensuremath{\mathtt{1}}\xspace}
\newcommand{\Neg}{\ensuremath{\mathtt{0}}\xspace}
\newcommand{\Non}{\ensuremath{\mathtt{\diamond}}\xspace}
\newcommand{\Rev}[1]{\ensuremath{\widetilde{#1}}}

\author{Dhananjay Ipparthi\thanks{IRIDIA, CoDE, Universit\'{e} libre de Bruxelles, Brussels, Belgium. \protect{\texttt{dhananjay.ipparthi@ulb.ac.be}}}
\and
Massimo Mastrangeli\thanks{Max Planck Institute for Intelligent Systems, Stuttgart, Germany. \protect{\texttt{mastrangeli@is.mpg.de}}}
\and
Andrew Winslow\thanks{D\'{e}partement d'Informatique, Universit\'{e} libre de Bruxelles, Brussels, Belgium. \protect{\texttt{awinslow@ulb.ac.be}}}
}

\title{Dipole Codes Attractively\\Encode Glue Functions}
\date{}

\begin{document}

\maketitle

\begin{abstract}
\emph{Dipole words} are sequences of magnetic dipoles, in which alike elements repel and opposite elements attract.
Magnetic dipoles contrast with more general sets of bonding types, called \emph{glues}, in which pairwise bonding strength is specified by a \emph{glue function}. 
We prove that every glue function~$g$ has a set of dipole words, called a \emph{dipole code}, that \emph{attractively encodes}~$g$: the pairwise attractions (positive or non-positive bond strength) between the words are identical to those of~$g$.
Moreover, we give such word sets of asymptotically optimal length.
Similar results are obtained for a commonly used subclass of glue functions.
\end{abstract}

\section{Introduction}
\label{sec:introduction}

Self-assembly is the autonomous organization of components into structures without supervision~\cite{Whitesides-2002}.
Here we consider controlling self-assembly using fixed arrangements of magnetic dipoles, specifically leveraging their ability to attract and repel according to their spatial configurations.
Some previous microscale~\cite{Hosokawa-1996,Bowden-2001,Clark-2001} and mesoscale~\cite{Bowden-1997,Bowden-2000,Rothemund-2000} self-assembling systems have used the capillary effects of surface tension as alternative bonding mechanisms.
More recently, molecular recognition of mesoscale components via surface chemistries has also been used~\cite{Cheng-2014,Harada-2010,Xiao-2015}.
However, magnets are among the most common sources of interaction force in micro- and mesoscale assembly, and they are used in both \emph{active} components that change bonding behavior~\cite{Bishop-2005,White-2004,White-2005} and in \emph{passive} components whose behavior is fixed~\cite{Bhalla-2010,Breivik-2001,Hosokawa-1994,Mermoud-2012}.

A primary limitation of dipole-based bonding is the limited number of interactions between dipoles: alike poles repel, while opposite poles attract. 
Frameworks by Bhalla et al.~\cite{Bhalla-2007,Bhalla-2010} and Majumder and Reif~\cite{Majumder-2008} describe an approach for obtaining more complex behaviors by arranging sequences of dipoles along the boundaries of components.
These \emph{dipole codes} are used to obtain many distinct bonding sites, called \emph{glues}, that interact only with a unique complementary glue.
Bhalla et al.~\cite{Bhalla-2007,Bhalla-2010,Bhalla-2012b} have demonstrated that dipole codes also work experimentally.

The use of \emph{DNA codes} is ubiquitous in DNA-based nanoscale self-assembly, where sequences of repeating nucleotides from the alphabet $\{\mathtt{A}, \mathtt{T}, \mathtt{C}, \mathtt{G}\}$ have been used experimentally to form dozens~\cite{Adleman-1994,Evans-2014,Rothemund-2004} or even hundreds~\cite{Ke-2012,Wei-2012} of glues.
Their theoretical study also is extensive, as seen in several surveys~\cite{Brenneman-2002,Garzon-2014,Mauri-2004,Montemanni-2015,Sager-2006}.
However, there are several practical aspects that differentiate the design of DNA codes from dipole codes.

For instance, the elasticity of DNA requires that codes must disallow multiple portions of a single code to bond~\cite{Kari-2005,Milenkovic-2006}, while chemistry requires that the codes must have balanced occurrences of letters $\mathtt{A}$, $\mathtt{T}$ and $\mathtt{G}$, $\mathtt{C}$~\cite{King-2003,Gaborit-2005,Montemanni-2014}. 
On the other hand, dipole codes have only a single pair of bonding letters, and inconsistency of mesoscale mixing allows codes to bond with even a single dipole pair~\cite{Bhalla-2010}.

Theoretical work on DNA-based systems has also demonstrated that the addition of a repelling force to systems with many glues increases computational power~\cite{Patitz-2011,Reif-2011} and efficiency~\cite{Doty-2013,Schweller-2013}.
Thus the construction of large numbers of glues with magnetic dipole sequences gives access to yet more techniques for controlling assembly.

\paragraph{Our contribution}

The frameworks of both Bhalla et al.~\cite{Bhalla-2007,Bhalla-2010} and Majumder and Reif~\cite{Majumder-2008} have several drawbacks.
First, neither formalizes how to obtain sets of magnetic dipole sequences that encode the behavior of a desired number of glues, nor glues with complex pairwise interactions.
Second, their sequences require precise control of the system's \emph{temperature}, the amount of force necessary for a bond to be irreversible, as well as additional component geometry.
Both are needed to prevent undesired bonds caused by dipole sequence pairs that only partially match or are misaligned.\footnote{Section 5.2 and Figure~11 of~\cite{Bhalla-2007} and Section~4 of~\cite{Majumder-2008} discuss these difficulties.}
Our work addresses both drawbacks, giving magnetic dipole sequences that encode pairwise bonding behaviors of arbitrarily many glues at fixed temperature and without misalignment.

Section~\ref{sec:definitions} begins by giving a formal model of magnetic dipole sequences, called \emph{dipole words}, and the net forces between them.
The pairwise bonding behaviors of a set of glues are defined by a \emph{glue function}, and we define what it means for a set of dipole words to \emph{encode} a glue function and thus be a \emph{dipole code}.
These definitions allow for the possibility of misaligned or weak bonds. 
They require that the dipole words work at fixed temperature (a pair can bond if the net attractive force is positive) and that all misaligned bonds are non-attractive.

Section~\ref{sec:encode-warmup} contains a ``warmup'' dipole word set construction that encodes \emph{canonical} glue functions, where each glue only bonds to itself or to a unique complementary glue. 
For any such function over $k$ glues, this construction gives a dipole code of length $O(k)$ that encodes it.

Section~\ref{sec:encode-general} improves this construction by extending it to all glue functions, allowing for \emph{flexible glues} (see~\cite{Aggarwal-2005}) that bond to many others.
Section~\ref{sec:encode-canonical} improves on length of dipole codes for canonical glue functions only, obtaining length-$O(\log{k})$ codes.
For both of these results, we also prove that the word lengths of the second and third constructions are asymptotically optimal.
Finally, Section~\ref{sec:conclusion} poses several remaining open problems.

\section{Definitions}
\label{sec:definitions}

\paragraph{Letters and words}
A \emph{letter} is a symbol $x$ in the \emph{alphabet} $\Sigma = \{\Neg, \Pos, \Non\}$. 
A \emph{dipole word} or simply a \emph{word} is a sequence of letters, and the \emph{length} of a word $W$, denoted $|W|$, is the number of letters in $W$.
For an integer $i \in \{1, 2, \dots, n\}$, $W[i]$ refers to the $i$th letter of $W$ and $W[-i]$ refers to the $i$th from the last letter of $W$.
The \emph{reverse} of a word $W$, written $\Rev{W}$, is the letters of $W$ in reverse order.

A \emph{subword} of $W$ is a contiguous sequence of letters in $W$. 
For integers $1 \leq i, j \leq |W|$ with $i \leq j$, $W[i..j]$ denotes the subword of $W$ from $W[i]$ to $W[j]$, inclusive.
As shorthand, $W[-i..] = W[-i..|W|]$ and $W[..j] = W[1..j]$.

\paragraph{Forces}
The \emph{force} of a pair of letters $x, y$ is defined by function 
\begin{displaymath}
f(x, y) = \left\{
\begin{array}{ll}
1 & : \{x, y\} = \{\Neg, \Pos \}\\
-1 & : \{x, y\} \in \{ \{\Neg \}, \{ \Pos \}\} \\
0 & : {\rm otherwise}
\end{array}
\right.
\end{displaymath}
and a pair of letters $x, y$ is \emph{attracted} or \emph{repelled} provided $f(x, y) = 1$ or $f(x, y) = -1$, respectively.
Similarly, for words $X$, $Y$ with $|X| = |Y|$, the \emph{force} of the pair $X, Y$, denoted $f(X, Y)$, is defined as 
$$f(X, Y) = \Sigma_{i = 1}^{|X|}f(X[i], Y[i])$$
and the pair $X, Y$ is \emph{attracted} provided $f(X, Y) > 0$.

\begin{figure}[ht]
\centering
\includegraphics[scale=1.2]{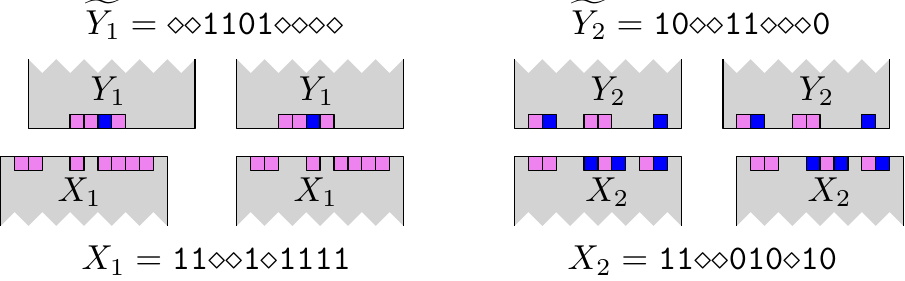}
\caption{
Left: the pair $X_1, \Rev{Y_1}$ is attracted, since $f(X_1, \Rev{Y_1}) = 1 > 0$, and aligned. 
Right: the pair $X_2, \Rev{Y_2}$ is not attracted, since $f(X_2, \Rev{Y_2}) = -3 + 2 = -1 \not > 0$, and not aligned, since $f(X_2[..-2], \Rev{Y_2}[2..]) = 3 > 0$.}
\label{fig:figure}
\end{figure}

\paragraph{Alignment}
For a pair $X, Y$ that corresponds to sequences of magnetic dipoles in self-assembling components, it may be the case that $X[i..]$ and $Y[..-i]$ are attracted, possibly causing unwanted and ``misaligned'' attachments.
A pair $X, Y$ with $|X| = |Y|$ are \emph{aligned} provided that for $A \in \{Y, \Rev{Y}\}$, $f(X, Y) \leq 0$ and:
$${\rm max}( \{  f(X[..i], A[-i..]), f(X[-i..], A[..i]) : i \in \mathbb{N}_{|X|-1} \}) \leq 0$$
The inclusion of both of $Y$ and $\Rev{Y}$ is to account for the possibility of two units attaching after being rotated ($A = \Rev{Y}$) or reflected ($A = Y$). 
The inclusion of $f(X, Y) \leq 0$ is to capture the fact that self-assembly systems often allow rotation without reflection, but not the converse. 

\paragraph{Glue functions}
Let $\mathbb{N}_k$ denote $\{1, 2, \dots, k\}$.
A \emph{$k$-glue function} is a function $g : \mathbb{N}_k^2 \rightarrow \mathbb{Z}$ such that $g(i, j) = g(j, i)$.
The \emph{bond graph} of a glue function $g$ is the graph $(V, E) = (\mathbb{N}_k, \{(i, j) : g(i, j) > 0\})$.
The following $k$-glue function is the \emph{canonical unsigned $k$-glue function}:
\begin{displaymath}
g(i, j) = \left\{
\begin{array}{ll}
1 & : i = j\\
0 & : {\rm otherwise}
\end{array}
\right.
\end{displaymath}
Similarly, the following $2k$-glue function is the \emph{canonical signed $2k$-glue function}:
\begin{displaymath}
g(i, j) = \left\{
\begin{array}{ll}
1 & : \{i, j\} = \{2a-1, 2a\}, a \in \mathbb{N}_k\\
0 & : {\rm otherwise}
\end{array}
\right.
\end{displaymath}

\paragraph{Encoding}
A set of common-length words $\mathcal{W} = W_1, W_2, \dots, W_k$ \emph{encodes} a $k$-glue function $g$ provided that for every $i, j \in \mathbb{N}_k$, $W_i, W_j$ are aligned and $f(W_i, \Rev{W_j}) = g(i, j)$. 
Similarly, $\mathcal{W}$ \emph{attractively encodes $g$} provided that for all $i, j$, $f(W_i, \Rev{W_j}) > 0$ if and only if $g(i, j) > 0$. 
Such an encoding set is called a \emph{dipole code} and the length of the words of the set is the code's \emph{length}.

\section{A First Encoding Result}
\label{sec:encode-warmup}

\begin{theorem}
\label{thm:encode-warmup}
For any $k$, there exists a length-$(6k+14)$ dipole code that attractively encodes the canonical signed $2k$-glue function.
\end{theorem}

\begin{proof}
The encoding dipole word set $\mathcal{W}$ consists of the following pairs for all $a \in \mathbb{N}_k$.
\begin{itemize}
\item $W_{2a-1} = \Non^{2k+5} \Pos \Pos (\Non \Pos)^{a-1} (\Pos \Non) (\Non \Pos)^{k-a} \Pos \Non^{2k+6}$
\item $W_{2a} = \Pos^{2k+6} \Neg (\Non \Non)^{k-a} (\Pos \Non) (\Non \Non)^{a-1} \Non \Non \Pos^{2k+5}$
\end{itemize}
Each word has length $6k+14$.
Establishing the encoding requires two steps: proving that every pair of words $W_i, W_j$ is aligned, and proving that $f(W_i, W_j) > 0$ if and only if $i = 2a-1, j = 2a$ for some $a \in \mathbb{N}_k$.

The very few \Neg letters mandate that words with positive force have almost no repelling letter pairs.
The long repeated sections of \Pos on both ends of $W_{2a}$ words, and \Pos letters in the centers of all words prevent misalignment by overlapping to form repelled letter pairs.
The pattern of \Pos and \Non letters in the middle of each word forms at least one repelled letter pair when aligned with any other word, except in complementary pairs $W_{2a-1}$ and $W_{2a}$, where the patterns have no force.

\paragraph{Alignment}
Let $W_i, W_j \in \mathcal{W}$.
It is easily observed that $f(W_i, W_j) \leq 0$.
Now suppose, for the sake of contradiction, that for some $l < |W_i| = |W_j|$ it is the case that $f(W_i[..l], W_j[-l..]) > 0$ or $f(W_i[..l], \Rev{W_j}[-l..]) > 0$.

It cannot be that both $i$ and $j$ are odd, since then $W_i$ and $W_j$ contain no \Neg letters and any subwords must have non-positive force.
Without loss of generality, assume $i$ is even.
Since the first and last $2k+6$ letters of $W_i$ are in $\{\Non, \Pos\}$, $l \geq 2k+7$ and thus $\Pos^{2k+6}$ is a subword of $W_i[..l]$.

If $j$ is even, then every subword of length $2k+6$ in $W_j$ contains at least two \Pos letters.
So the subword $\Pos^{2k+6}$ of $W_i[..l]$ must be involved in at least two repelled letter pairs.
Also, $W_i$ and $W_j$ contain two \Neg letters total, so any pair of subwords have at most two attracted letter pairs.
So the force between any pair of subwords of $W_i$, $W_j$ is at most $2-2 = 0$.

If $j$ is odd, then every subword of $W_j$ with length $2k+6$ has at least one \Pos letter, except $W[-(2k+6)..]$.
Then since $l \geq 2k+7$, the subword $\Pos^{2k+6}$ of $W_i[..l]$ is involved in a repelled letter pair. 
Also, $W_i$ and $W_j$ contain one \Neg letter total, so any pair of subwords have force at most $1-1=0$.

\paragraph{Force}
Clearly $f(W_i, \Rev{W_j}) \leq 0$ if $i$, $j$ have the same parity.
So all that remains is to verify that for every $a, b \in \mathbb{N}_k$, it is the case that $f(W_{2a-1}, \Rev{W_{2b}}) > 0$ if and only if $a = b$.
Observe that $W_{2a-1}[2k+8..4k+7] = (\Non \Pos)^{a-1} (\Pos \Non) (\Non \Pos)^{k-a}$ and similarly $\Rev{W_{2b}}[2k+8..4k+7] = (\Non \Non)^{b-1} (\Non \Pos) (\Non \Non)^{k-b}$.
The remaining portions of $W_{2a-1}$ and $\Rev{W_{2b}}$ have force $f(\Non^{2k+5} \Pos \Pos, \Pos^{2k+5} \Non \Non) + f(\Pos \Non^{2k+6}, \Neg \Pos^{2k+6}) = 1$.
So for all $a, b \in \mathbb{N}_k$, it follows that $f(W_{2a-1}, \Rev{W_{2b}}) = 1$ if $a = b$ and $0$ otherwise.
\end{proof}

A set of $2k$ dipole words encoding the canonical signed $2k$-glue function can be concatenated in pairs to yield $k$ dipole words encoding the canonical $k$-glue function.
For instance, the pairs $W_1, W_2$ and $W_3, W_4$ of length $6\cdot 2 + 14 = 26$ that together attractively encode the canonical signed $4$-glue function yield two words $W_1 \Non^{26} W_2$ and $W_3 \Non^{26} W_4$ encoding the canonical unsigned $2$-glue function. 
More generally, for larger $k$, concatenating the pairs $W_{2a-1}, W_{2a}$ into $W_{2a-1} \Non^{6k+14} W_{2a}$ for all $a \in \mathbb{N}_k$ yields the canonical unsigned $k$-glue function.

\begin{corollary}
\label{cor:encode-warmup}
For any $k$, there exists a length-$(18k+42)$ dipole code that attractively encodes the canonical unsigned $k$-glue function.
\end{corollary}

\section{Encoding General Glue Functions}
\label{sec:encode-general}

In following with the previous section, we obtain a result for bipartite bond graphs (Theorem~\ref{thm:encode-general}) and then use this to prove a matching result for general bond graphs (Corollary~\ref{cor:noncanonical}).

\begin{theorem}
\label{thm:encode-general}
For any $k$-glue function $g$ with a bipartite bond graph, there exists a length-$(3k + 14)$ dipole code that attractively encodes $g$.
\end{theorem}

\begin{proof}
Since the bond graph is bipartite, it can be written as $(V_1, V_2, E) = (\mathbb{N}_{k'}, \mathbb{N}_k \setminus \mathbb{N}_{k'}, E)$.
Assume without loss of generality that $|V_1| \geq |V_2|$ and so $|V_2| \leq k/2$.

Let $\alpha : \mathbb{Z} \rightarrow \{\Non \Pos, \Pos \Non\}$ with 
\begin{displaymath}
\alpha(n) = \left\{
\begin{array}{ll}
\Non \Pos & : n \leq 0\\
\Pos \Non & : n > 0
\end{array}
\right.
\end{displaymath}

For integers $a \in V_1$ and $b \in V_2$, define $M_a = \bigcup_{b=k'+1}^{|V|}{\alpha(g(a, b))}$, i.e. $M_a$ is the concatenation of $|V_2|$ 2-letter words encoding the neighbors of $a$ in $V_2$.
Notice that $|M_a| = 2|V_2| \leq 2(k/2) = k$.
The dipole word set $\mathcal{W}$ consists of the following set of dipole words for all $a \in V_1$, $b \in V_2$:
\begin{itemize}
\item $W_a = \Non^{k+5} \Pos \Pos M_a \Pos \Non^{k+6}$
\item $W_b = \Pos^{k+6} \Neg (\Non \Non)^{|V_2|-b} \Pos \Non (\Non \Non)^{b-|V_1|-1} \Non \Non \Pos^{k+5}$
\end{itemize}
The set $\mathcal{W}$ attractively encodes $g$ by extending the approach in Section~\ref{sec:encode-warmup} to allow half of the words, namely $W_a$, to have multiple ``complementary'' words.
The words are shorter here due to a reduction in the number of words (from $2k$ to $k$).

Following the alignment portion of the proof of Theorem~\ref{thm:encode-warmup} establishes that every pair is aligned and only pairs $W_i$, $W_j$ with $i = a \in V_1, j = b \in V_2$ can be attracted.
Observe that $f(M_a, (\Non \Non)^{b-|V_1|-1} \Pos \Non (\Non \Non)^{|V_2|-b}) = 0$ if $g(a, b) > 0$ and $-1$ otherwise.
Then since $f(\Non^{k+5} \Pos \Pos, \Pos^{k+5} \Non \Non) + f(\Pos \Non^{k+6}, \Neg \Pos^{k+6}) = 1$, $f(W_a, \Rev{W_b}) > 0$ if and only if $g(a, b) = 1$.
So a pair $W_a, W_b$ is attracted if and only if $g(a, b) > 1$.
\end{proof}

For an unsigned $k$-glue function $g$, consider the dipole code $\mathcal{W}$ constructed by Theorem~\ref{thm:encode-general} that attractively encodes a bipartite $2k$-glue function $g'$, where:
\begin{displaymath}
g'(i, j) = \left\{
\begin{array}{ll}
g(i, j-k) & : i \leq k < j\\
g(i-k, j) & : i > k \geq j\\
0 & {\rm otherwise}
\end{array}
\right.
\end{displaymath}
Concatenating signed pairs $W_i, W_{i+k}$ for all $i \in \mathbb{N}_k$ into words $W_i \Non^{3k+14} W_j$ attractively encodes the original unsigned $k$-glue function $g$, since the word $W_i \Non^{3k+14} W_{i+k}$ is attracted to a word $W_j \Non^{3k+14} W_{j+k}$ if and only if $g(i, j) > 0$.

\begin{corollary}
\label{cor:noncanonical}
For any $k$-glue function $g$, there exists a length-$(9k+42)$ dipole code that attractively encodes $g$.
\end{corollary}

As it turns out, dipole codes of $O(k)$-length is the best possible. 

\begin{theorem}
Most $k$-glue functions can only be attractively encoded with dipole codes of length $\Omega(k)$.
\end{theorem}

\begin{proof}
Observe that all $k$-glue functions can be partitioned into equal-size sets, where two functions $g$, $g'$ are in the same set if and only if for all $i, j \in \mathbb{N}_k^2$, $g(i, j) = g'(i, j)$ or $g(i, j) = -g'(i, j)+1$ for all $i, j$.
That is, two functions are in the same set if they output identical or opposite strengths for all inputs.
There are $2^{\binom{k}{2}+k}$ functions in each set, each specified by a distinct sequence of $\binom{k}{2}+k$ choices for whether to output a positive or non-positive value for each input pair $i, j$ with $i \leq j$.
 
Thus any method of encoding $k$-glue functions (including as encoding dipole word sets) requires $\log(2^{\binom{k}{2}+k}) = \Omega(k^2)$ bits to specify at least half of these functions. 
Encoding $\Omega(k^2)$ bits in a dipole word set of $k$ words requires that the words contain $\Omega(k^2/k) = \Omega(k)$ bits each, and so have length $\Omega(k)$.
\end{proof}

\section{Encoding Canonical Glue Functions}
\label{sec:encode-canonical}

In the previous section, we adapted the construction of Section~\ref{sec:encode-warmup} to attractively encode general $k$-glue functions using dipole codes optimal length $\Theta(k)$. 
Here we do the same for the special class of canonical glue functions, improving the length of the dipole code to $O(\log{k})$.
This is easily seen to be optimal, as there must be at least $k$ words in the code.

\begin{theorem}
For any $k \in \mathbb{N}$ divisible by~4, there exists a set of length-$(20k+3)$ words that that attractively encodes the canonical signed $2\binom{k}{k/2}$-glue function.
\end{theorem}

\begin{proof}
Let $\alpha, \beta : \{0, 1\} \rightarrow \Sigma^2$ be functions defined by 
\begin{displaymath}
\alpha(n) = \left\{
\begin{array}{ll}
\Pos \Non & : n = 0\\
\Non \Pos & : n = 1
\end{array}
\right\}
\;\;\;\;
\beta(n) = \left\{
\begin{array}{ll}
\Pos \Non & : n = 0\\
\Neg \Non & : n = 1
\end{array}
\right\}
\end{displaymath}
Let $\mathcal{C} = \{C_1, C_2, \dots, C_{\binom{k}{k/2}}\}$ be the set of all bit words of length $k$ with equal numbers of~0's and~1's.
Let $M_{\alpha, C} = \bigcup_{i=1}^k{\alpha(C_i)}$ and $M_{\beta, C} = \bigcup_{i=1}^k{\beta(C_i)}$.
Define a \emph{bit} word to be a word over the alphabet $\{0, 1\}$.
We claim the following set $\mathcal{W}$ of dipole words attractively encodes the canonical signed $2\binom{k}{k/2}$-glue function.
\begin{itemize}
\item $W_{2a-1} = \Non^{6k} \Pos^{k/4} \Non^{11k/4} \Pos \Pos M_{\alpha, C_a} \Pos \Non^{11k/4+1} \Pos^{k/4-1} \Non^{6k+1}$
\item $W_{2a} = \Pos^{9k+1} \Non M_{\beta, \Rev{C_a}} \Non \Non \Pos^{9k}$
\end{itemize}
These words form a ``compressed'' version of the dipole code in Section~\ref{sec:encode-warmup} by changing the middle word patterns from a unary-based encoding with no force to a binary encoding with large, positive force.
The regions of \Pos and \Non letters at both ends of the words create a balancing set of repelling letter pairs that only fail to exceed the positive force of the middle patterns if these patterns match. 

\paragraph{Alignment}
Let $W_i, W_j \in \mathcal{W}$.
If $i$, $j$ have the same parity, then $f(W_i, W_j) \leq -k/2$.
For any $l < |W_i| = |W_j|$, $f(W_i[..l], W_j[-l..])$ is either less than~1 (if $i$, $j$ are odd) or less than $-9k+(2k+3) + k/2 < 0$ (if $i$, $j$ are even), and the same holds for $f(W_i, \Rev{W_j}[-l..])$.

So consider the remaining case: $i = 2a-1$, $j = 2b$ for $a, b \in \mathbb{N}_{\binom{k}{k/2}}$ (and $l < |W_i|$).
For $l \leq 15k$ and $A \in \{W_{2b}, \Rev{W_{2b}}\}$, it is the case that $f(W_{2a-1}[..l], A[-l..]), f(W_{2a-1}[-l..], A[..l]) \leq 0$, since the \Neg letters of $A$ (that do not appear in the first or last $9k$ letters of $A$) coincide only with \Non letters of $W_{2a-1}[..l]$ or $W_{2a-1}[-l..]$ (the first and last $6k$ letters of $W_{2a-1}$).

For $18k \geq l > 15k$, since the subword $\Pos^{9k}$ is in $A$ and every subword $W_{2a-1}$ with length $9k$ excluding the last contains at least $k/4$ \Pos letters, there are at least $k/4$ repelled letter pairs involving letters in $A$.
Moreover, the portion of $M_{\beta, \Rev{C_b}}$ in $A$ coincides with a subword of $W_{2a-1}$, namely a subword of $\Non^{6k} \Pos^{k/4} \Non^{11k/4}$ or $\Non^{11k/4+1} \Pos^{k/4-1} \Non^{6k+1}$, that contains at most $k/4$ \Pos letters.
So it must be that $f(W_{2a-1}[..l], A[-l..]), f(W_{2a-1}[-l..], A[..l]) \leq -k/4 + k/4 = 0$.

For $l > 18k$, $A$ contains non-overlapping subwords $\Pos^{9k}$ and $\Pos^{18k-9k-(2k+3)} = \Pos^{7k-3}$, a superword of $\Pos^{6k + k/4}$. 
Then since every subword of $W_{2a-1}$ with length $6k + k/4$ contains $k/4$ \Pos letters, there are at least $2(k/4)=k/2$ repelling letter pairs involving letters in $A$.
Moreover, $W_{2a-1}$ and $A$ contain a total of at most $k/2$ \Neg letters.
Thus $f(W_{2a-1}[..l], A[-l..]), f(W_{2a-1}[-l..], A[..l]) \leq -k/2 + k/2 = 0$.

\paragraph{Force}
Clearly $f(W_i, \Rev{W_j}) \leq 0$ for every $i$, $j$ with the same parity. 
So consider $f(W_{2a-1}, \Rev{W_{2b}})$ for $a, b \in \mathbb{N}_{\binom{k}{k/2}}$.
This force is equal to $k/2-1 + f(M_{\alpha, C_a}, \Rev{M_{\beta, \Rev{C_b}}})$.
Inspecting $\alpha, \beta$, it is seen that $f(M_{\alpha, C_a}, \Rev{M_{\beta, \Rev{C_b}}}) = k/2$ provided $a = b$, and at most $k/2-2$ otherwise.
Thus $f(W_{2a-1}, \Rev{W_{2b}}) = 1$ provided $a = b$ and non-positive otherwise. 
\end{proof}

Applying algebra and a similar transformation as done for Corollary~\ref{cor:encode-warmup} gives results for both signed and unsigned $k$-glue functions in terms of $k$:

\begin{corollary}
\label{cor:encode-hard}
For any $k$, there exists a length-$O(\log{k})$ dipole code that attractively encodes the canonical unsigned $k$-glue or signed $2k$-glue function.
\end{corollary}

\section{Conclusion}
\label{sec:conclusion}

All of the encodings here are attractive encodings, and thus the word-word forces do not precisely match the outputs of the encoded glue functions.
Removing the ``attractive'' modifier is of primary importance, starting with canonical glue functions. 

\begin{openproblem}
\label{op:encode}
Can dipole codes encode canonical glue functions? Glue functions with no negative outputs? All glue functions?
\end{openproblem}

The existence of dipole codes that encode the glue functions 
\begin{displaymath}
g(i, j) = \left\{
\begin{array}{ll}
-1 & : i = j = 1\\
1 & : i = j > 1\\
0 & : {\rm otherwise} 
\end{array}
\right.
\end{displaymath}
imply that systems of square components using magnet sequences embedded along their edges are capable of universal computation, following from results of Patitz, Schweller, and Summers~\cite{Patitz-2011}.

Prior work by Bhalla et al.~\cite{Bhalla-2010} uses geometry to prevent unaligned bonding.
Such a restriction eases the difficulty of designing dipole codes, and yields an easier version of Open Problem~\ref{op:encode}.

\begin{openproblem}
Can dipole codes encode all glue functions if words are not required to be aligned?
\end{openproblem}

Forthcoming work by the authors demonstrates that dipole codes can be physically implemented.
The implementation consists of centimeter-sized, 3D-printed components orbitally stirred and bond with one another via sequences of magnets embedded into their faces.
The components have dimensions $19~{\rm mm} \times 19~{\rm mm} \times 5~{\rm mm}$ (see Figure~\ref{fig:physical}).
Neodymium disc magnets (N45, $1~{\rm mm}$ diameter, $1~{\rm mm}$ length) are embedded into 9~cylindrical recesses ($1~{\rm mm}$ diameter, $2~{\rm mm}$ length) equally spaced along $19~{\rm mm} \times 5~{\rm mm}$ faces.

\begin{figure}[h!]
  \centering
  \setlength\fboxsep{0pt}
  \setlength\fboxrule{.5pt}
  \subfloat[]{
    \raisebox{10mm}{\includegraphics[height=35mm]{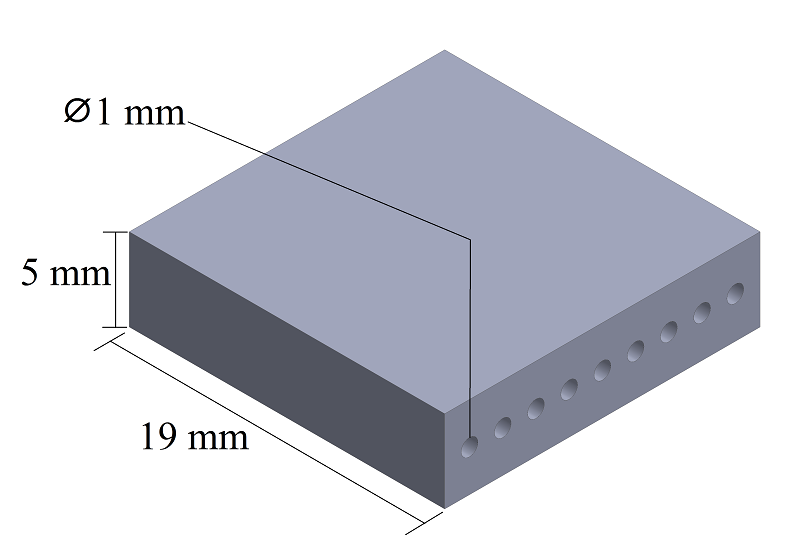}}
    \label{fig:cad_drawing}}~
  \subfloat[]{
    \includegraphics[height=50mm]{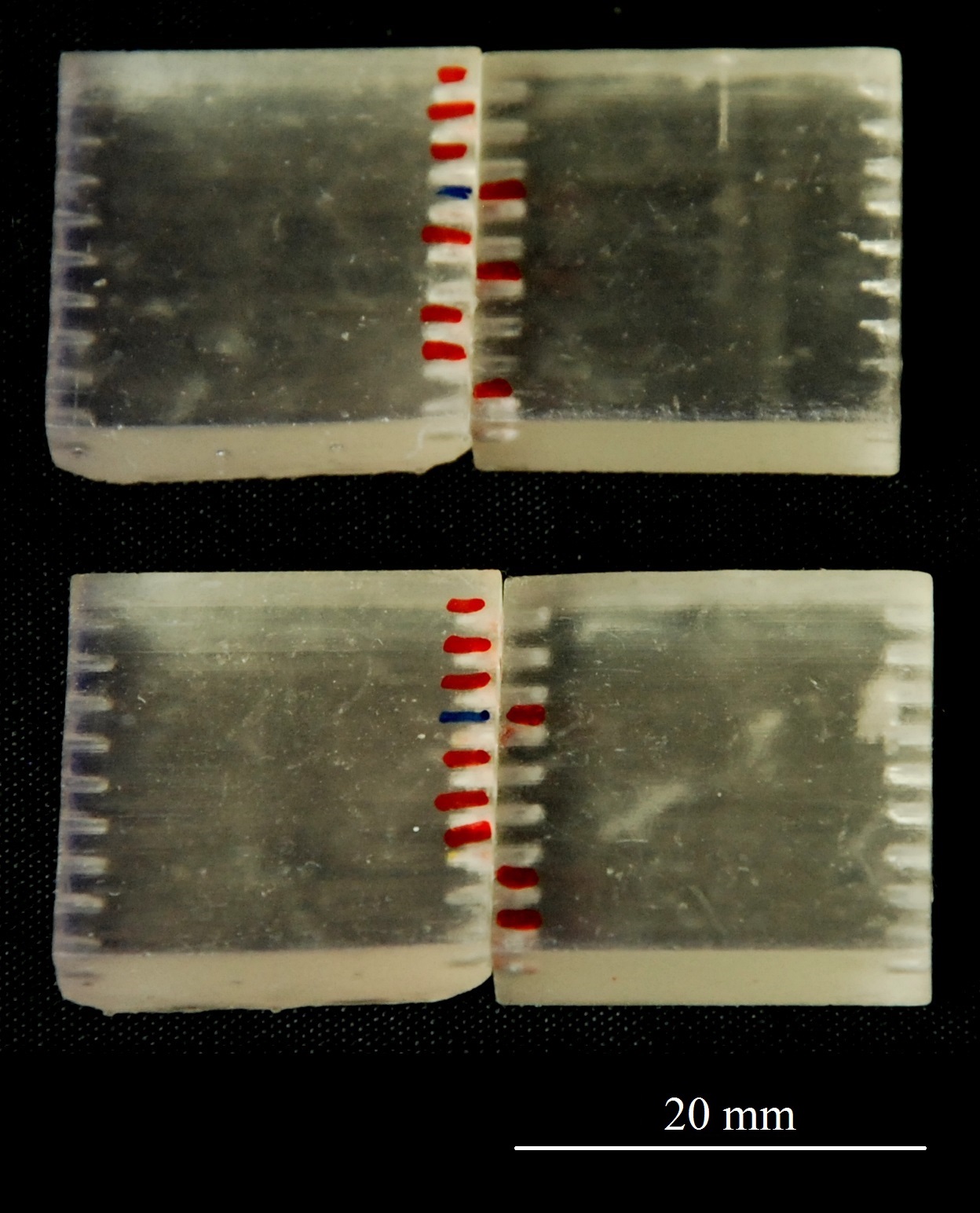}
    \label{fig:component}}
  \caption{(a) Component dimensions. (b) Component pairs bonding along faces containing complementary sequences of embedded magnets. Blue and red marks indicate magnets embedded with north and south poles oriented outwards, respectively.}
  \label{fig:physical}
\end{figure}

The recess and magnet lengths leave an \emph{air gap} between the magnet and face.
The air gap and recess spacing enforce that magnets interact pairwise as in the model: magnets embedded in the same face are sufficiently far apart that no magnet can strongly interact with more than one magnet in the same face. 
Magnet embedding sequences are determined by dipole words: a \Pos letter is a magnet with positive pole oriented outwards, a \Neg letter is a magnet negative pole oriented outwards, and a \Non letter is no magnet.
The implementation uses a length-9 code that encodes the canonical signed 4-glue function: $\Pos \Pos \Pos \Neg \Pos \Non \Pos \Pos \Non$, $\Pos \Non \Non \Pos \Non \Pos \Non \Non \Non$, $\Pos \Pos \Pos \Neg \Pos \Pos \Pos \Non \Non$, $\Pos \Pos \Non \Non \Non \Pos \Non \Non \Non$.
Preliminary experiments demonstrate that components reliably bond irreversibly via aligned complementary dipole words, while misaligned or non-complentary bonds rarely form and are easily broken by the kinetic energy imparted by the orbital stirring.

While the codes developed in this work are asymptotically optimal with small constants, scaling physical implementations to the sizes needed to use even small instances of these codes is challenging due to increased component weight, agitation forces, and assembly time.
Thus computing codes as close to optimal as possible is key to achieving physical implementations.

\begin{openproblem}
Can canonical signed $2k$-glue functions be attractively encoded by dipole codes of length $\lceil \log_3(2k) \rceil + o(\log{k})$? 
\end{openproblem}

\begin{openproblem}
What is the computational complexity of the following problem: given a glue function $g$ and integer $l$, does there exist a dipole code of length at most $l$ (attractively) encoding $g$?
\end{openproblem}

\newpage

\bibliographystyle{plain}
\bibliography{dipole}

\begin{thebibliography}{10}
\expandafter\ifx\csname url\endcsname\relax
  \def\url#1{\texttt{#1}}\fi
\expandafter\ifx\csname urlprefix\endcsname\relax\def\urlprefix{URL }\fi
\expandafter\ifx\csname href\endcsname\relax
  \def\href#1#2{#2} \def\path#1{#1}\fi

\bibitem{Whitesides-2002}
G.~M. Whitesides, B.~Grzybowski, Self-assembly at all scales, Science 295
  (2002) 2418--2421.

\bibitem{Hosokawa-1996}
K.~Hosokawa, I.~Shimoyama, H.~Miura, Two-dimensional micro-self-assembly using
  the surface tension of water, Sensors and Actuators A: Physical 57~(2) (1996)
  117--125.

\bibitem{Bowden-2001}
N.~Bowden, F.~Arias, T.~Deng, G.~M. Whitesides, Self-assembly of microscale
  objects at a liquid/liquid interface through lateral capillary forces,
  Langmuir 17 (2001) 1757--1765.

\bibitem{Clark-2001}
T.~D. Clark, J.~Tien, D.~C. Duffy, K.~E. Paul, G.~M. Whitesides, Self-assembly
  of 10-$\mu$m-sized objects into ordered three-dimensional arrays, Journal of
  the American Chemical Society 123 (2001) 7677--7682.

\bibitem{Bowden-1997}
N.~Bowden, A.~Terfort, J.~Carbeck, G.~M. Whitesides, Self-assembly of mesoscale
  objects into ordered two-dimensional arrays, Science 276~(5310) (1997)
  233--235.

\bibitem{Bowden-2000}
N.~Bowden, S.~R.~J. Oliver, G.~M. Whitesides, Mesoscale self-assembly: 
  capillary bonds and negative menisci, Journal of Physical Chemistry {B}
  104~(12) (2000) 2714--2724.

\bibitem{Rothemund-2000}
P.~W.~K. Rothemund, Using lateral capillary forces to compute by self-assembly,
  Proceedings of the National Academy of Sciences 97~(3) (2000) 984--989.

\bibitem{Cheng-2014}
M.~Cheng, F.~Shi, J.~Li, Z.~Lin, C.~Jiang, M.~Xiao, L.~Zhang, W.~Yang,
  T.~Nishi, Macroscopic supramolecular assembly of rigid building blocks
  through a flexible spacing coating, Advanced Materials 26 (2014) 3009--3013.

\bibitem{Harada-2010}
A.~Harada, R.~Kobayashi, Y.~Takashima, A.~Hashidzume, H.~Yamaguchi, Macroscopic
  self-assembly through molecular recognition, Nature Chemistry 3 (2010)
  34--37.

\bibitem{Xiao-2015}
M.~Xiao, Y.~Xian, F.~Shi, Precise macroscopic supramolecular assembly by
  combining spontaneous locomotion driven by the marangoni effect and molecular
  recognition, Angewandte Chemie International Edition 54~(31) (2015)
  8952--8956.

\bibitem{Bishop-2005}
J.~Bishop, S.~Burden, E.~Klavins, R.~Kreisberg, Programmable parts: a
  demonstration of the grammatical approach to self-organization, in:
  Proceedings of International Conference on Intelligent Robots and Systems
  2005, 2005, pp. 3684--3691.

\bibitem{White-2004}
P.~J. White, K.~Kopanski, H.~Lipson, Stochastic self-reconfigurable cellular
  robotics, in: Proceedings of the International Conference on Robotics and
  Automation 2004, 2004, pp. 2888--2893.

\bibitem{White-2005}
P.~J. White, V.~Zykov, J.~Bongard, H.~Lipson, Three dimensional stochastic
  reconfiguration of modular robots, in: Proceedings of Robotics: Science and
  Systems {I}, 2005, pp. 161--168.

\bibitem{Bhalla-2010}
N.~Bhalla, P.~J. Bentley, C.~Jacob, Evolving physical self-assembling systems
  in two-dimensions, in: ICES 2010, Vol. 6274 of LNCS, Springer, 2010, pp.
  381--392.

\bibitem{Breivik-2001}
J.~Breivik, Self-organization of template-replicating polymers and the
  spontaneous rise of genetic information, Entropy 3 (2001) 273--279.

\bibitem{Hosokawa-1994}
K.~Hosokawa, I.~Shimoyama, H.~Miura, Dynamics of self-assembling systems:
  Analogy with chemical kinetics, Artificial Life 1~(4) (1995) 413--427.

\bibitem{Mermoud-2012}
G.~Mermoud, M.~Mastrangeli, U.~Upadhyay, A.~Martinoli, Real-time automated
  modeling and control of self-assembling systems, in: Proceedings of 2012 IEEE
  International Conference on Robotics and Automation (ICRA), 2012, pp.
  4266--4273.

\bibitem{Bhalla-2007}
N.~Bhalla, P.~J. Bentley, C.~Jacob, Mapping virtual self-assembly rules to
  physical systems, in: Proceedings of the 6th International Conference on
  Unconventional Computation, 2007, pp. 117--147.

\bibitem{Majumder-2008}
U.~Majumder, J.~H. Reif, A framework for designing novel magnetic tiles capable
  of complex self-assemblies, in: Unconventional Computing 2008, Vol. 5204 of
  LNCS, Springer, 2008, pp. 129--145.

\bibitem{Bhalla-2012b}
N.~Bhalla, P.~J. Bentley, P.~D. Vize, C.~Jacob, Programming and evolving
  physical self-assembling systems in three dimensions, Natural Computing
  11~(3) (2012) 475--498.

\bibitem{Adleman-1994}
L.~Adleman, Molecular computation of solutions to combinatorial problems,
  Nature 266~(5187) (1994) 1021--1024.

\bibitem{Evans-2014}
C.~G. Evans, Crystals that count! physical principles and experimental
  investigations of {DNA} tile self-assembly, Ph.D. thesis, Caltech (2014).

\bibitem{Rothemund-2004}
P.~W.~K. Rothemund, N.~Padadakis, E.~Winfree, Algorithmic self-assembly of
  {DNA} sierpinski triangles, PLoS Biology 2~(12) (2004) 2041--2053.

\bibitem{Ke-2012}
Y.~Ke, L.~Ong, W.~Shih, P.~Yin, Three-dimensional structures self-assembled
  from {DNA} bricks, Science 338 (2012) 1177--1183.

\bibitem{Wei-2012}
B.~Wei, M.~Dai, P.~Yin, Complex shapes self-assembled from single-stranded
  {DNA} tiles, Nature 485 (2012) 623--626.

\bibitem{Brenneman-2002}
A.~Brenneman, A.~Condon, Stranddesign for biomolecular computation, Theoretical
  Computer Science 287~(1) (2002) 39--58.

\bibitem{Garzon-2014}
M.~H. Garzon, {DNA} codeword design: Theory and applications, Parallel
  Processing Letters 24~(2) (2014) 1440001.

\bibitem{Mauri-2004}
G.~Mauri, C.~Ferretti, Word design for molecular computing: A survey, in:
  J.~Chen, J.~Reif (Eds.), DNA 9, Vol. 2943 of LNCS, Springer, 2004, pp.
  37--47.

\bibitem{Montemanni-2015}
R.~Montemanni, Combinatorial optimization algorithms for the design of codes: a
  survey, Journal of Applied Operational Research 7~(1) (2015) 36--41.

\bibitem{Sager-2006}
J.~Sager, D.~Stefanovic, Designing nucleotide sequences for computation: A
  survey of constraints, in: A.~Carbone, N.~A. Pierce (Eds.), DNA 11, Vol. 3892
  of LNCS, Springer, 2006, pp. 275--289.

\bibitem{Kari-2005}
L.~Kari, S.~Konstantinidis, P.~Sos\'{i}k, Preventing undesirable bonds between
  {DNA} codewords, in: C.~Ferretti, G.~Mauri, C.~Zandron (Eds.), DNA 10, Vol.
  3384 of LNCS, Springer, 2005, pp. 182--191.

\bibitem{Milenkovic-2006}
O.~Milenkovic, N.~Kashyap, On the design of codes for {DNA} computing, in:
  O.~Ytrehus (Ed.), WCC 2005, Vol. 3969 of LNCS, Springer, 2006, pp. 100--119.

\bibitem{King-2003}
O.~D. King, Bounds for {DNA} codes with constant {GC}-content, Electronic
  Journal of Combinatorics 10~(1) (2003) R33.

\bibitem{Gaborit-2005}
P.~Gaborit, O.~D. King, Linear constructions for {DNA} codes, Theoretical
  Computer Science 334~(1--3) (2005) 99--113.

\bibitem{Montemanni-2014}
R.~Montemanni, D.~H. Smith, N.~Koul, Three metaheuristics for the construction
  of constant {GC}-content {DNA} codes, in: ICAOR 2014, Vol.~6 of LNMS,
  Springer, 2014, pp. 167--175.

\bibitem{Patitz-2011}
M.~J. Patitz, R.~T. Schweller, S.~M. Summers, Exact shapes and turing
  universality at temperature 1 with a single negative glue, in: DNA 2011, Vol.
  6937 of LNCS, Springer, 2011, pp. 175--189.

\bibitem{Reif-2011}
J.~H. Reif, S.~Sahu, P.~Yin, Complexity of graph self-assembly in accretive
  systems and self-destructible systems, Theoretical Computer Science 412
  (2011) 1592--1605.

\bibitem{Doty-2013}
D.~Doty, L.~Kari, B.~Masson, Negative interactions in irreversible
  self-assembly, Algorithmica 66~(1) (2013) 153--172.

\bibitem{Schweller-2013}
R.~T. Schweller, M.~Sherman, Fuel efficient computation in passive
  self-assembly, in: Proceedings of the 24th Annual Symposium on Discrete
  Algorithms, 2013, pp. 1513--1525.

\bibitem{Aggarwal-2005}
G.~Aggarwal, Q.~Cheng, M.~Goldwasser, M.~Kao, P.~de~Espanes, R.~Schweller,
  Complexities for generalized models of self-assembly, SIAM Journal on
  Computing 34~(6) (2005) 1493--1515.

\end{thebibliography}

\end{document}